\documentclass[11pt,letter]{article}

\usepackage{amsmath}
\usepackage{xspace}
\usepackage{amssymb}
\usepackage{epsfig}
\usepackage{graphicx}
\usepackage{bbold}
\usepackage[usenames,dvipsnames]{color}

\newtheorem{Thm}{Theorem}[section]
\newtheorem{Lem}[Thm]{Lemma}
\newtheorem{Cor}[Thm]{Corollary}

\newtheorem{Claim}{Claim}

\newtheorem{Def}{Definition}

\newtheorem{Conj}{Conjecture}
\newenvironment{proof}{\noindent {\textbf{Proof }}}{$\Box$ \medskip}

\def\QEDopen{{\setlength{\fboxsep}{0pt}\setlength{\fboxrule}{0.2pt}\fbox{\rule[0pt]{0pt}{1.3ex}\rule[0pt]{1.3ex}{0pt}}}}
\def\QED{\QEDopen}
\def\proof{\noindent{\bf Proof}: }
\def\endproof{\hspace*{\fill}~\QED\par\endtrivlist\unskip}

\newcommand {\ie} {\textit{i.e.}\xspace}
\newcommand {\st} {\textit{s.t.}\xspace}

\newcommand\av{\mbox{\bf{\bf E}}}

\newcommand\ket[1]{| #1 \rangle}
\newcommand\bra[1]{\langle #1 |}
\newcommand\qip[2]{\langle #1 | #2 \rangle}

\begin{document}

\title{\textbf{BQP} and \textbf{PPAD}}

\author{Yang Daniel Li}

\maketitle
\thispagestyle{empty}
\setcounter{page}{0}
\abstract{
We initiate the study of the relationship between two complexity classes, \textbf{BQP} (\textbf{B}ounded-Error \textbf{Q}uantum \textbf{P}olynomial-Time) and \textbf{PPAD} (\textbf{P}olynomial \textbf{P}arity \textbf{A}rgument, \textbf{D}irected). We first give a conjecture that \textbf{PPAD} is contained in \textbf{BQP}, and show a necessary and sufficient condition for the conjecture to hold. Then we prove that the conjecture is not true under the oracle model. In the end, we raise some interesting open problems/future directions.
}

\newpage

\section{Introduction}
Quantum computing and algorithmic game theory are two exciting and active areas in the last two decades. Quantum computing lies in the intersection of computer science and quantum physics, and studies the power and limitation of a quantum computer. Algorithmic game theory touches upon the foundations of both computer science and economics, and aims to design efficient algorithms in strategic circumstances. If we want to come up with some examples that are able to demonstrate the interaction between computer science and other disciplines, then quantum computing and algorithmic game theory are two perfect candidates. Quantum mechanics may provide additional computational power, and quantum computers can test the foundations of quantum mechanics. Economics lends some strategic views, and computer science rewards with computational points of view. We refer readers to \cite{NC2000} and \cite{NRTV2007} for more information.

The central topics of quantum computing and algorithmic game theory are the hardness of two complexity classes, \textbf{BQP} (\textbf{B}ounded-Error \textbf{Q}uantum \textbf{P}olynomial-Time) and \textbf{PPAD} (\textbf{P}olynomial \textbf{P}arity \textbf{A}rgument, \textbf{D}irected). {\bf BQP}, as introduced by Bernstein and Vazarani \cite{BV1997}, characterizes efficient computation of a quantum computer and is the quantum analog of {\bf BPP}. Very little is known about {\bf BQP}, and a wide belief is that {\bf BQP} and {\bf NP} are incomparable complexity classes \cite{BBBV1997, BV1997, Aar2010}. Papadimitriou introduced the complexity class {\bf PPAD} \cite{Pap1994}, which is a special class between {\bf P} and {\bf NP}. Since then, the hardness of {\bf PPAD} has also become a longstanding open problem. Although lots of important problems, say the problem of computing a Nash equilibrium (\emph{NASH} for short) \cite{DGP2009, CDT2009}, were shown to be {\bf PPAD}-complete, there have been very few relations from {\bf PPAD} to other complexity classes.

In this paper, we initiate the study of the relationship between {\bf BQP} and {\bf PPAD}. The representative problem of {\bf PPAD} is \emph{NASH} \cite{Pap1994}, and the most well-known problem in {\bf BQP} is factoring \cite{Sho1997}. Both \emph{NASH} and factoring are in the complexity class {\bf TFNP} (the set of total search problems, see \cite{MP1991}) in the sense that every instance of \emph{NASH} and factoring always has a solution. Therefore, it seems that there may be some relationship between {\bf PPAD} and {\bf BQP}.

In fact, this possible relation was (implicitly) mentioned in a talk given by Papadimitriou ten years ago \cite{Pap2001}. Papadimitrious said that ``together with factoring, the complexity of finding a Nash equilibrium is in my opinion the most important concrete open question on the boundary of {\bf P} today". In other words, Papadimitriou asked: do there exist (classically and deterministically) efficient algorithms for factoring and \emph{NASH}? As it have been shown that there exist efficient quantum algorithms for factoring \cite{Sho1997}, it is natural for us to ask: do there exist efficient quantum algorithms for \emph{NASH}? More generally, is {\bf PPAD} contained in {\bf BQP}?

Our conjecture is that {\bf PPAD} is contained in {\bf BQP}. Formally,
\begin{Conj}
{\bf PPAD} $\subseteq$ {\bf BQP}.
\end{Conj}
The conceived relationship can be illustrated by Figure $1$, where the \emph{green+red} is {\bf BQP} and the \emph{red} denotes {\bf PPAD}.

\begin{figure*}[t]
\centering
\includegraphics[angle=0, width=0.7\textwidth]{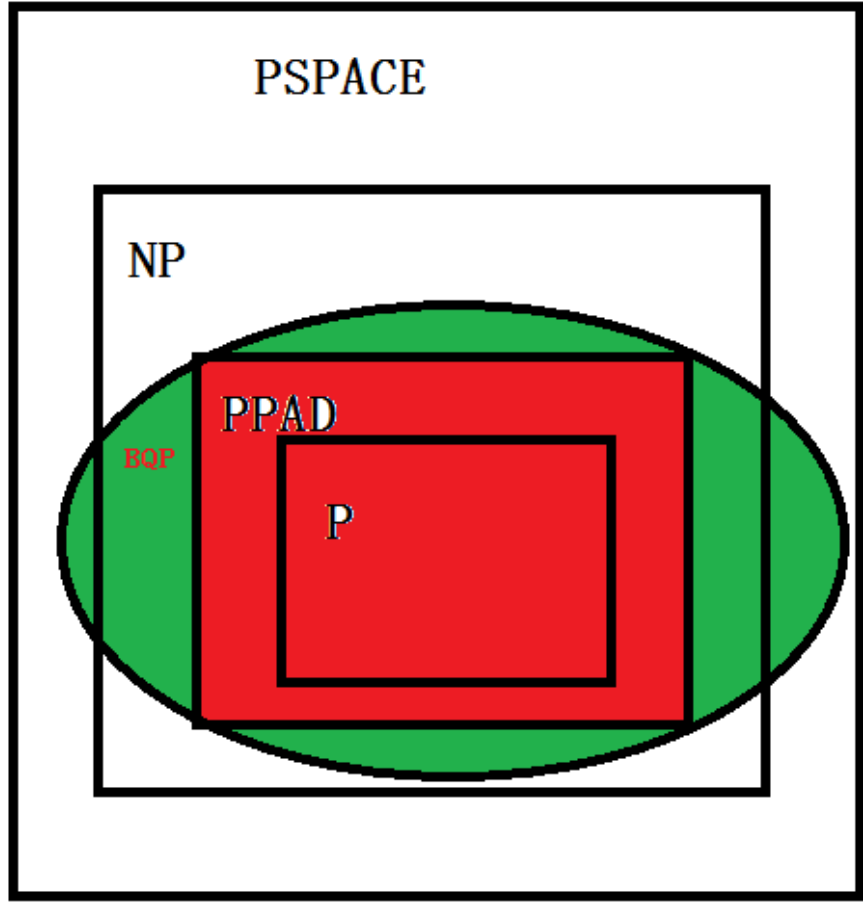}
\vspace{-5pt}
\caption{\label{Figure 1}The conceived picture}
\vspace{-5pt}
\label{fig:exp:bra}
\end{figure*}

We will formally define quantum Nash equilibrium, the quantum analog of Nash equilibrium, and prove the fact that {\bf PPAD} is contained in {\bf BQP} if and only if there exists a polynomial-time quantum algorithm for computing a quantum Nash equilibrium. Therefore, to prove the conjecture, we need to find an efficient (polynomial-time) quantum algorithm, and to disprove the conjecture, we have to show a super-polynomial lower bound for the time complexity of computing a quantum Nash equilibrium.

Another way to express the conjecture is that quantum computers can compute a Nash equilibrium in polynomial time, or that quantum computers can exponentially speed-up the computation of a Nash equilibrium. And we will rule out this possibility under the oracle model.

The organization of this paper is as follows. Section $2$ presents the definition of {\bf BQP} and {\bf PPAD}, and in Section $3$, we introduce the notion of quantum Nash equilibrium and analyze it. Section $4$ provides the necessary and sufficient condition for our conjecture to hold. Section $5$ proves a lower bound of computing a Nash equilibrium using quantum computers under the oracle model. And we concludes the paper with some open problems/future directions in Section $6$.

\section{Preliminaries}
\subsection{Notation}
Some notations used throughout the paper are listed here.
\begin{itemize}
\item $\mathbb{N}$: the set of natural numbers, $\{1, 2, 3, \ldots \}$.
\item $[n]$: the integer set $\{1, 2, \ldots, n\}$.
\item $\mathbb{R}$: the set of real numbers.
\item $||\phi||$: the $2$-norm of a vector $\phi$. If $\phi$ is a quantum state $\sum_x \alpha_x \ket{x}$, then $||\phi|| = \sqrt{\sum_x |\alpha_x|^2}$.
\end{itemize}

\subsection{BQP}
\cite{BV1997} introduced the notion of {\bf BQP}, and a simplified version is as follows.
\begin{Def}
A language $L$ is in {\bf BQP} if and only if there exists a polynomial-time uniform family of quantum circuits $\{Q_n: n \in \mathbb{N}\}$, such that
\begin{itemize}
\item For all $n \in \mathbb{N}$, $Q_n$ takes $n$ qubits as input and outputs $1$ bit.
\item For all $x$ in $L$, $Pr(Q_{|x|}(x) = 1) \ge 2/3$.
\item For all $x$ not in $L$, $Pr(Q_{|x|}(x) = 0) \ge 2/3$.
\end{itemize}
\end{Def}

\subsection{PPAD and PPAD-completeness}
Total search problems are problems for which solutions are guaranteed to exist, and the challenge is to find a specific solution. In \cite{Pap1994}, Papadimitriou defined the following total search problem.
\begin{Def} (END-OF-THE-LINE)
Let $S$ (standing for {\bf S}uccessor) and $P$ (standing for {\bf P}redecessor) be two polynomial size circuits that given input strings $\{0,1\}^n$ output strings $\{0,1\}^n$. We further require that $P(0^n)=0^n\ne S(0^n)$. The aim is to find an input $x$ such that $P(S(x)) \ne x$ or $S(P(x))\ne x \ne 0$.
\end{Def}
A more intuitive description of END-OF-THE-LINE is as follows. $G$ is a (possibly exponentially large) directed graph with no isolated vertices, and with every vertex having at most one predecessor and one successor. $G$ is specified by giving a polynomial-time computable function $f(v)$ (polynomial in the size of $v$) that returns the predecessor and successor (if they exist) of the vertex $v$. Given a vertex $s$ in $G$ with no predecessor, find a vertex $t \ne s$ with no predecessor or no successor. (The input to the problem is the source vertex $s$ and the function $f(v)$). In other words, we want any source or sink of the directed graph other than $s$.

\textbf{PPAD} was defined based on this problem.
\begin{Def} (\textbf{PPAD})
The complexity class {\bf PPAD} contains all total search problems reducible to END-OF-THE-LINE in polynomial time.
\end{Def}
\textbf{PPAD}-completeness was also defined.
\begin{Def} (\textbf{PPAD}-completeness)
A problem is called {\bf PPAD}-complete if it is in {\bf PPAD} and all problems in {\bf PPAD} can reduce to it in polynomial time.
\end{Def}

\section{Quantum Nash Equilibrium}
\subsection{Classical Equilibria}
First, we review classical Nash equilibria and correlated equilibria, all of which can be found in \cite{NRTV2007}.

In a classical game there are $n$ players, labeled $\{1,2,\ldots,n\}$. Each player $i$ has a set $S_i$ of strategies. 
We use $s=(s_1,\ldots, s_n)$ to denote the vector of strategies selected by the players and $S=\times_i S_i$ to denote the set of all possible joint strategies. Each player $i$ has a utility function $u_i: S \rightarrow \mathbb{R}$, giving the payoff or utility $u_i(s)$ to player $i$ on the joint strategy $s$. There is a solution concept called Nash equilibrium, in which the equilibrium strategies are known by all players, and no player can gain more by unilaterally modifying his or her choice. Formally,
\begin{Def}
A \emph{mixed Nash equilibrium} is a probability vector $p = p_1 \times \ldots \times p_n$ for some probability distributions $p_i$'s over $S_i$'s satisfying that
\begin{align*}
	\sum_{s_{-i}} p_{-i}(s_{-i}) u_i(s_i,s_{-i}) \geq  \sum_{s_{-i}} p_{-i}(s_{-i}) u_i(s_i',s_{-i}), \qquad \forall i\in [n], \forall s'_i\in S_i,\forall  s_i \in S_i \; s.t. \; p_i(s_i)>0,
\end{align*}
\end{Def}
where $s_{-i}$ is the strategies chosen by players but player $i$, and $p_{-i}$ denotes the probability distribution over $s_{-i}$. Informally speaking, for a mixed Nash equilibrium, the expected payoff over probability distribution of $s_{-i}$ is maximized, i.e. $\av_{s_{-i}}[u_i(s_i,s_{-i})] \ge \av_{s_{-i}}[u_i(s_i',s_{-i})]$. We can further relax the Nash condition and define an $\epsilon$-approximate Nash equilibrium to be a profile of mixed strategies such that no player can gain more than $\epsilon$ amount by changing his/her own strategy unilaterally. Formally,
\begin{Def}
An \emph{$\epsilon$-approximate Nash equilibrium} is a probability vector $p = p_1 \times \ldots \times p_n$ for some probability distributions $p_i$'s over $S_i$'s satisfying that
\begin{align*}
	\sum_{s_{-i}} p_{-i}(s_{-i}) u_i(s_i,s_{-i}) \geq  \sum_{s_{-i}} p_{-i}(s_{-i}) u_i(s_i',s_{-i})-\epsilon, \qquad \forall i\in [n], \forall s'_i\in S_i,\forall  s_i \in S_i \; s.t. \; p_i(s_i)>0,
\end{align*}
\end{Def}
where $\epsilon > 0$. In addition, the probability distribution of each player may not be independent, but correlated, forming the notion of correlated equilibria.
\begin{Def} \label{thm:CNE}
A \emph{correlated equilibrium} is a probability distribution $p$ over $S$ satisfying that
\begin{align*}
	\sum_{s_{-i}} p(s_i,s_{-i}) u_i(s_i,s_{-i}) \geq  \sum_{s_{-i}} p(s_i,s_{-i}) u_i(s_i',s_{-i}), \qquad \forall i\in [n], \forall s_i, s'_i\in S_i.
\end{align*}
\end{Def}
Notice that a correlated equilibrium $p$ is a Nash equilibrium if and only if $p$ is a product distribution.
\subsection{Quantum Equilibria}
This part generalizes classical equilibria to quantum equilibria, where players are allowed to use ``quantum" strategies. To be more precise, each player $i$ now has a Hilbert space $H_i = span\{s_i: s_i\in S_i\}$, and the joint strategy can be any quantum state $\rho$ in $H = \otimes_i H_i$. The payoff/utility for player $i$ on joint strategy $\rho$ is $\mu_i(\rho) =\av[u_i(s(\rho))]=\sum_s \bra{s} \rho \ket{s} u_i(s)$, where $s(\rho)$ is the outcome pure strategy when $\rho$ is measured according to the computational basis $\{s: s\in S\}$. Note that what each player $i$ can do is to apply an admissible super-operator $\Phi_i$ on her own space $H_i$. We sometimes write $\Phi_i$ for $\Phi_i \otimes I_{-i}$. We use $CPTP(X)$ to denote the set of all admissible (completely positive and trace preserving) super-operators on a space $X$. The notions of quantum Nash equilibria and quantum correlated equilibria are defined as follows.
\begin{Def}
A \emph{quantum Nash equilibrium} is a quantum strategy $\rho = \rho_1 \otimes \ldots \otimes \rho_n$ for $\rho_i$'s in $H_i$'s satisfying
\begin{align*}
	\sum_s \bra{s} \rho \ket{s} u_i(s) \ge \sum_s \bra{s} \Phi_i(\rho) \ket{s} u_i(s), \qquad \forall i\in [n], \forall \Phi_i \in CPTP(H_i)
\end{align*}
\end{Def}
\begin{Def} \label{thm:QCNE}
A \emph{quantum correlated equilibrium} is a quantum strategy $\rho$ in $H$ satisfying
\begin{align*}
	\sum_s \bra{s} \rho \ket{s} u_i(s) \ge \sum_s \bra{s} \Phi_i(\rho) \ket{s} u_i(s), \qquad \forall i\in [n], \forall \Phi_i \in CPTP(H_i)
\end{align*}
\end{Def}

\subsection{Relations between Classical and Quantum Equilibria}
This section studies the relation between classical and quantum equilibria. A quantum mixed state $\rho$ naturally induces a classical distribution $p$ over $S$ defined by
\begin{equation}
	p(s) = \rho_{ss}
\end{equation}
While taking diagonal entries seems to be the most natural mapping from quantum states to  classical distributions, there are more options for the mapping in other direction. Given a classical distribution $p$ over $S$, we can consider
\begin{enumerate}
	\item $\rho(p) = \sum_s p(s) \ket{s}\bra{s}$,
	\item $\ket{\psi(p)} = \sum_s \sqrt{p(s)} \ket{s}$, or
	\item any density matrix $\rho$ with $p(s) = \rho_{ss}$ satisfied.
\end{enumerate}
We want to study whether equilibria in one world, classical or quantum, implies equilibria in the other world. The following theorem says that quantum always implies classical.
\begin{Thm}\label{thm: QE to CE}
	If $\rho$ is a quantum correlated equilibrium, then $p$ defined by $p(s) = \rho_{ss}$ is a classical correlated equilibrium. In particular, if $\rho$ is a quantum Nash equilibrium, then $p$ is a classical Nash equilibrium.
\end{Thm}
\proof See the appendix. \endproof

The implication from classical to quantum is much more complicated. The following theorem says that the first mapping always gives a quantum equilibrium. That is, the utility of $i$ cannot be increased for a classical equilibrium even when player $i$ is allowed to have quantum operations.
\begin{Thm}\label{thm: CE to QE}
	If $p$ is a classical correlated equilibrium, then $\rho = \Sigma_{s\in S} \ket{s}\bra{s}$ is a quantum correlated equilibrium. In particular, if $p$ is a classical Nash equilibrium, then $\rho$ as defined is a quantum Nash equilibrium.
\end{Thm}
\proof See the appendix. \endproof
Following this result, we are able to give an affirmative answer to an important problem, whether quantum Nash equilibria always exist.

\begin{Cor}\label{thm: QNE-exist}
	For all standard game $G$ with finite number of players and strategies, quantum Nash equilibria always exist.
\end{Cor}
\proof See the appendix. \endproof

The second way of inducing a quantum state is interesting: It preserves (uncorrelated) Nash equilibria, but does not preserve correlated Nash equilibria in general.
\begin{Thm}\label{thm: purestate}
There exists a classical correlated equilibrium $p$ with $\ket{\psi} = \sum_{s} \sqrt{p(s)}\ket{s}$ not being a quantum correlated equilibrium. However, if $p$ is a classical Nash equilibrium, then $\ket{\psi} = \sum_{s} \sqrt{p(s)}\ket{s}$ is a quantum Nash equilibrium.
\end{Thm}
\proof See the appendix. \endproof
Finally, for the third mapping, \ie a general $\rho$ with $p(s) = \rho_{ss}$ satisfied, the equilibrium property can be heavily destroyed, even if $p$ is uncorrelated. (Actually, we will show such counterexamples even for two-player symmetric games.)
\begin{Thm}\label{thm: examples}
	There exist $\rho$ and $p$ satisfying that $p(s)= \rho_{ss}$, $p$ is a classical Nash equilibrium, but $\rho$ is not even a quantum \emph{correlated} equilibrium.
\end{Thm}
\proof See the appendix. \endproof
Despite of the above fact, one should not think that the large range of the third type of mappings always enables some mapping to destroy the equilibria.
\begin{Thm}\label{thm: examples2}
	There exist classical correlated equilibria $p$, such that all quantum states $\rho$ with $\rho_{ss} = p(s)$ are quantum correlated equilibria.
\end{Thm}
\proof See the appendix. \endproof

\section{A Necessary and Sufficient Condition}

In this section, we prove the following theorem.

\begin{Thm}\label{thm: QNE-PPAD}
${\bf PPAD} \subseteq {\bf BQP}$ if and only if there exists a polynomial-time quantum algorithm for finding a quantum Nash equilibrium.	
\end{Thm}
\begin{proof}

If ${\bf PPAD} \subseteq {\bf BQP}$, then for a game $G$, there exists a polynomial-time quantum algorithm for finding a Nash equilibrium $p$, since finding a Nash equilibrium is a ${\bf PPAD}$-complete problem. As shown in the proof of Corollary \ref{thm: QNE-exist}, we can always convert $p$ to a quantum equilibrium $\rho$ in polynomial time. Hence, there exists a polynomial-time quantum algorithm for finding a quantum Nash equilibrium.	

Next we will prove the inverse direction for the statement.

We define a new problem as follows.

\begin{Def} \label{Sample-Nash}
    \textit{SAMPLE-NASH} is a search problem that, on input a game $G$, outputs a pure strategy $s$ sampled from a fixed Nash equilibrium $p$ of the game $G$.
\end{Def}

Suppose that we are given a game $G$ and we find a quantum Nash equilibrium $\rho$ in polynomial-time using the quantum algorithm. Here we assume that the induced classical probability distribution induced from $\rho$ is $p$, defined by $p(s)=\rho_{ss}$. By measuring $\rho$ according to the computational basis $\{s: s\in S\}$, we can obtain a pure strategy $s$, which is sampled according to $p$. According to Theorem \ref{thm: QE to CE}, $p$ is a classical Nash equilibrium, and therefore $s$ is an output for \textit{SAMPLE-NASH}. So we are able to obtain the output for \textit{SAMPLE-NASH} in polynomial time. Now we have a polynomial-time quantum algorithm for \textit{SAMPLE-NASH}.

We have the following result, which is to be proved later.

\begin{Lem} \label{sample-PPAD}
A ${\bf PPAD}$-complete problem can be reduced to \textit{SAMPLE-NASH} in randomized polynomial time.
\end{Lem}

Therefore, we have a polynomial-time quantum algorithm for a ${\bf PPAD}$-complete problem, implying ${\bf PPAD} \subseteq {\bf BQP}$.

\end{proof}

\subsection{Proof of Lemma \ref{sample-PPAD}}

To prove Lemma \ref{sample-PPAD}, we use the following result.

\begin{Lem} \cite{CDT2009} \label{approx-PPAD}\\
    For any constant $c>0$, the problem of computing a $1/m^c$-approximate Nash equilibrium of a positively normalized \footnote{In \cite{CDT2009}, the game matrices are normalized in the sense that all the entries are between $0$ and $1$ (positively normalized), or between $-1$ and $1$. } $m\times m$ bimatrix game is {\bf PPAD}-complete.
\end{Lem}

We just need to reduce the problem in Lemma \ref{approx-PPAD}, to \textit{SAMPLE-NASH} in randomized polynomial time.

For an instance of the problem in Lemma \ref{approx-PPAD}, namely a positively normalized $m\times m$ bimatrix game $G$, we use $G$ as the input for \textit{SAMPLE-NASH}. We assume that we have an algorithm $A$ for \textit{SAMPLE-NASH}, and we want to use $A$ to construct an algorithm for the problem in Lemma \ref{approx-PPAD} in randomized polynomial time. Suppose the output of $A$ is sampled from a Nash equilibrium $p=p_1\times p_2$.

\begin{Lem} \label{Lemma 1}
Suppose that $p = p_1 \times p_2$ is a Nash equilibrium of a positively normalized $m\times m$ bimatrix game $G$, and that the output of A, an algorithm for \textit{SAMPLE-NASH}, is sampled from $p$. For any $\epsilon = O(1/m^c)$, with high probability, we will get a probability distribution $q = q_1\times q_2$ with $||q_1-p_1||_1\le \epsilon$ and $||q_2-p_2||_1\le \epsilon$, after running A for $O(m^2\epsilon^{-2})$ times.
\end{Lem}

\begin{Lem} \label{Lemma 2}
     Suppose that $p = p_1 \times p_2$ is a Nash equilibrium of a positively normalized $m\times m$ bimatrix game $G$. Any probability distribution $q = q_1 \times q_2$ with $||q_1-p_1||_1\le \epsilon$ and $||q_2-p_2||_1\le \epsilon$, is a $2\epsilon$-approximate Nash equilibrium of game $G$.
\end{Lem}

By Lemma \ref{Lemma 1}, we can run algorithm A for $O(m^2\epsilon^{-2})$ times to construct a desired probability distribution $q$ with high probability. By Lemma \ref{Lemma 2}, $q$ is an $2\epsilon$-approximate Nash equilibrium. To find an $1/m^c$-approximate Nash equilibrium, we need to use $A$ for $O(m^{2c+2})$ times, which is polynomial in input size $2m^2$.

\subsubsection{Proof of Lemma \ref{Lemma 1}}

We assume that player $1$'s $m$ strategies are $s_1, s_2,\ldots, s_m$. Define $k=\lceil 4000m^2/\epsilon^2 \rceil = O(m^2\epsilon^{-2})$. For each $i\in [k]$, $j\in [m]$, define random variable $X_{ij}$ taking values in $\{0,1\}$, where $X_{ij}=1$ with probability $p_1(s_j)$.

Suppose that $\epsilon_j=\frac{\epsilon}{2m}$ for each $j\in [m]$. Define random variables $X_j$ to be $X_j=\frac{\sum_{i\in [k]} X_{ij}}{k}$ for each $j\in [m]$. By Chernoff bound,
\begin{equation}
Pr( X_j \ge p_1(s_j)+\epsilon_j) \le e^{-2\epsilon_j^2 k}
\end{equation}
and
\begin{equation}
Pr( X_j \le p_1(s_j)-\epsilon_j) \le e^{-2\epsilon_j^2 k}.
\end{equation}
Define a probability vector $q_1$ to be $(X_1, \ldots, X_m)$, which is a distribution over strategies $(s_1, s_2, \ldots, s_m)$. It is easily checkable that $\sum_j X_j = 1$ and that
\begin{eqnarray}
Pr(||q_1-p_1||_1 \le \sum_{j\in [m]}|X_j-p_1(s_j)| \le \sum_{j\in [m]}\epsilon_j \le \epsilon) &\ge& 1- \sum_{j\in [m]}2e^{-2\epsilon_j^2 k}\\
&\ge& 0.995.
\end{eqnarray}

Similarly, we can get $q_2$ satisfying $||q_2-p_2||_1\le \epsilon$ with probability at least $0.995$. By union bound, we can get the desired $q$ with probability at least $0.99$.

\subsubsection{Proof of Lemma \ref{Lemma 2}}

For all $i$ in $\{1,2\}$, for all $s_i'$ in the set of strategies of player $i$, and for all $s_i$ in the support of the set of strategies of player $i$, we have the following:

\begin{align*}
& \sum_{s_{-i}}q_{-i}(s_{-i})u(s_i's_{-i})-\sum_{s_{-i}}q_{-i}(s_{-i})u(s_is_{-i})\\
= &  \sum_{s_{-i}}p_{-i}(s_{-i})u(s_i's_{-i})+\sum_{s_{-i}}(q_{-i}(s_{-i})-p_{-i}(s_{-i}))u(s_i's_{-i})-\sum_{s_{-i}}q_{-i}(s_{-i})u(s_is_{-i}) \\
\le &  \sum_{s_{-i}}p_{-i}(s_{-i})u(s_is_{-i})+ ||q_{-i}-p_{-i}||_1\max_su(s)-\sum_{s_{-i}}q_{-i}(s_{-i})u(s_is_{-i}) \\
\le &  \sum_{s_{-i}}p_{-i}(s_{-i})u(s_is_{-i})+ \epsilon \times 1-\sum_{s_{-i}}q_{-i}(s_{-i})u(s_is_{-i}) \\
\le &  \sum_{s_{-i}}(p_{-i}(s_{-i})-q_{-i}(s_{-i}))u(s_is_{-i})+ \epsilon \\
\le &  ||q_{-i}-p_{-i}||_1 \max_su(s)+ \epsilon \\
\le &  \epsilon \times 1 + \epsilon \\
\le &  2\epsilon
\end{align*}

Following the definition of approximate Nash equilibrium, $q$ is a $2\epsilon$-approximate Nash equilibrium of $G$.

\section{A Lower Bound under the Oracle Model}
\subsection{The Oracle Model}
The oracle model is also called black-box model, or relativized model, and is one of simplest models in computer science. Suppose that there is a boolean function $f: [N] \rightarrow \{0, 1\}$, and that $f$ can be computed in polynomial time. We want to find an $x \in [N]$, such that $f(x) = 1$. In the context of {\bf PPAD}-complete problems, $N$, which could be exponential in the size of the input, is the number of points in the search space, and $f(x) = 1$ for $x \in [N]$ means that $x$ is the answer we desire. For a {\bf PPAD}-complete problem, there always exists an $x \in [N]$, such that $f(x) = 1$, and the question is that we do not know where it is. Such an $f$ is called an oracle, and we want to compute an $x \in [N]$, such that $f(x) = 1$.

In an oracle model, algorithms are allowed to make queries to the oracle but are prohibited to take advantage of what underlies the oracle. Since we use quantum algorithms here, we can also make use of quantum superposition. For instance, for a quantum state $\sum_x \alpha_x \ket{x}$, we first add some ancilla qubits, obtaining $\sum_x \alpha_x \ket{x}\ket{0}$, and then make a single query to the oracle, getting $\sum_x \alpha_x \ket{x} \ket{f(x)}$.

In summary, in the oracle model, the function $f$ can be seen as the input, and we need to design a quantum algorithm to find a solution $x \in [N]$ with $f(x) = 1$, which is guaranteed to exist. The time complexity is what we care and is defined to the number of queries made to the oracle.

\subsection{The Lower Bound}
We use a hybrid argument of \cite{BBBV1997, Vaz2004} to show a result when there is only one $x \in [N]$ such that $f(x) = 1$, namely for the problems with a single solution. Hybrid argument is from a classic paper by Yao \cite{Yao1982}, and later has numerous applications in cryptography and complexity theory \cite{BM1984, GL1989, HILL1999, INW1994, Nis1991, Nis1992, NW1994}. Thus, our proof is not new, and existing techniques are enough to prove the result. This is partly due to the fact that the oracle model is very well-studied.

\begin{Thm}
Under the oracle model, to solve a {\bf PPAD}-complete problem with a single solution, any quantum algorithm has to make at least $\Omega(\sqrt{N})$ queries to the oracle.
\end{Thm}

\proof

Suppose $A$ is an (arbitrary) algorithm under the oracle model, and it makes $k$ queries to the input oracle. If $k = \Omega(N)$, then everything is done and we need to do nothing. So a reasonable assumption is that $k = o(N)$.

We define an auxiliary oracle function $h: [N] \rightarrow \{0, 1\}$ with $h(y) = 0$ for all $y \in [N]$. Such an oracle cannot characterize any {\bf PPAD}-complete problem, as there are always solutions for {\bf PPAD}-complete problems while here $h$ means no solution at all. So we use $h$ just purely for analysis.

Run $A$ on $h$ and we call such a run $A_h$. Let $\sum_{y: y \in [N]} \alpha_{y, t}\ket{y}$ be the query at time $t \in [k]$, and let the query magnitude of $y$ to be $\sum_{t \in [k]} |\alpha_{y, t}|^2$. It is not hard to see that the expected query magnitude over all possible $y$ is $E_{y} (\sum_t |\alpha_{y, t}|^2) = k / N$. We have the following claim.
\begin{Claim}\label{claim}
There exist $z_1, z_2 \in [N]$ with $z_1 \ne z_2$, such that $\sum_t |\alpha_{z_1, t}|^2 \le (k + 1)/ N$ and $\sum_t |\alpha_{z_2, t}|^2 \le (k + 1)/ N$.
\end{Claim}
By Cauchy-Schwartz inequality, we know that $\sum_t|\alpha_{z_1, t}| \le (k + 1)/\sqrt{N}$ and that $\sum_t|\alpha_{z_2, t}| \le (k + 1)/\sqrt{N}$.

Let $\phi_{h, t}$, $t\in [k]$ be the states of $A_h$ after the $t$-th step. We define two oracles $g_1: [N] \rightarrow \{0, 1\}$ and $g_2: [N] \rightarrow \{0, 1\}$:
\begin{itemize}
\item $g_1(z_1) = 1$ and for all $y \ne z_1$, $g_1(y) = 0$;
\item $g_2(z_2) = 1$ and for all $y \ne z_2$, $g_2(y) = 0$.
\end{itemize}
$g_1$ and $g_2$ are the legal inputs of $A$ and correspond to {\bf PPAD}-complete problems. Now run the algorithm $A$ on $g_1$ (the run is denoted as $A_{g_1}$) and suppose the final state of $A_{g_1}$ is $\phi_{g_1, k}$. By hybrid argument, we have the following claim.
\begin{Claim}\cite{Vaz2004}\\
$\phi_{h, k} - \phi_{g_1, k} = \sum_{t = 1}^{k}E_t$, where $||E_t|| \le \sqrt{2}|\alpha_{z_1, t}|$.
\end{Claim}

Along with the triangle inequality, we have
\begin{eqnarray}
||\phi_{h, k} - \phi_{g_1, k}|| &\le& \sum_t||E_t|| \nonumber \\
&\le& \sqrt{2} \sum_t|\alpha_{z_1,t}| \nonumber \\
&\le& (k+1)\sqrt{2/N}.
\end{eqnarray}
Similarly, if we run the algorithm $A$ on $g_2$ (the run is denoted as $A_{g_2}$) and the final state of $A_{g_2}$ is $\phi_{g_2, k}$, then a hybrid argument and the triangle inequality could show that
\begin{equation}
||\phi_{h, k} - \phi_{g_2, k}|| \le (k + 1)\sqrt{2/N}.
\end{equation}
If we apply the triangle inequality for another time, we get
\begin{equation}
||\phi_{g_1, k} - \phi_{g_2, k}|| \le 2(k + 1)\sqrt{2/N},
\end{equation}
implying that $\phi_{g_1, k}$ and $\phi_{g_2, k}$ can be distinguished with probability at most $O(k/\sqrt{N})$. Since $z_1 \ne z_2$, if $A$ can solve problems corresponding to $g_1$ and $g_2$, namely if $A$ can find $z_1$ and $z_2$, it should at least distinguish $g_1$ and $g_2$, and also $\phi_{g_1, k}$ and $\phi_{g_2, k}$ with some constant probability. As a result, $A$ should at least make $\Omega(\sqrt{N})$ queries.

\endproof

When there are multiple solutions, say $p$ solutions, then $k = \Omega(\sqrt{N/p})$, which is a straightforward generalization from the theorem above. More formally,
\begin{Cor}
Under the oracle model, to solve a {\bf PPAD}-complete problem with $p$ solutions, any quantum algorithm has to make at least $\Omega(\sqrt{N/p})$ queries to the oracle.
\end{Cor}

\subsubsection{Proof of Claim \ref{claim}}
Let us suppose that there does not exist $z_1, z_2 \in [N]$ with $z_1 \ne z_2$, such that $\sum_t |\alpha_{z_1, t}|^2 \le (k + 1)/ N$ and $\sum_t |\alpha_{z_2, t}|^2 \le (k + 1)/ N$. This means there is at most one $z \in [N]$ such that $\sum_t |\alpha_{z, t}|^2 \le (k + 1)/ N$, and for all $y \ne z$, $y \in [N]$, $\sum_t |\alpha_{y, t}|^2 > (k + 1)/ N$. Thus,
\begin{eqnarray}
\sum_{y: y \in [N]} \sum_t |\alpha_{y, t}|^2 &=&  \sum_{y: y \ne z} \sum_t |\alpha_{y, t}|^2 + \sum_t |\alpha_{z, t}|^2 \nonumber \\
&\ge& \sum_{y: y \ne z} \sum_t |\alpha_{y, t}|^2 \nonumber \\
&>& (N - 1) \times (k + 1)/ N \nonumber \\
&=& k + 1 - (k + 1)/N \nonumber \\
&>& k. \label{contra1}
\end{eqnarray}

But we have already known that
\begin{equation}
E_{y: y \in [N]} (\sum_t |\alpha_{y, t}|^2) = k / N,
\end{equation}
and that
\begin{equation}
\sum_{y: y \in [N]} (\sum_t |\alpha_{y, t}|^2) = k. \label{contra2}
\end{equation}
The inequality (\ref{contra1}) and the equation (\ref{contra2}) exhibit clear contradiction. Consequently, our assumption that there does not exist $z_1, z_2 \in [N]$ with $z_1 \ne z_2$, such that $\sum_t |\alpha_{z_1, t}|^2 \le (k + 1)/ N$ and $\sum_t |\alpha_{z_2, t}|^2 \le (k + 1)/ N$ is incorrect. This completes the proof of Claim \ref{claim}.

\section{Concluding Remarks}
On the one hand, it seems that the well-studied oracle model presents us an insurmountable obstacle towards an exponentially speed-up using quantum computers for computing {\bf PPAD}-complete problems. If we want to make a step closer to prove our conjecture that {\bf PPAD} is contained in {\bf BQP}, we have to get rid of oracles and design new structures that can provide more information. We believe that this may need fundamental revolution in the field of quantum computing. The theory community has spent lots of effort in designing quantum algorithms for factoring as well as graph isomorphism, two special problems between {\bf P} and {\bf NP}. And now it is the time that we turn our attention to the third special problem, \emph{NASH}, or more generally {\bf PPAD}-complete problems.

On the other hand, it seems that purely exploiting the potential of quantum superposition is not enough, and quantum entanglement may play a more important role as a resource for quantum computation. It is well-known that quantum information theory relies on entanglement in two quite different contexts: as a resource for quantum computation and as a source for nonlocal correlations among different parties. It is strange and not understood that entanglement is crucially linked with nonlocality but not with computation. Quantum computation and nonlocality are two faces of entanglement, and more connections should be established in the future.

\section{Acknowledgments}
Thanks to Shengyu Zhang for discussions at the early stage of this work.

\newpage
\thispagestyle{empty}

\bibliography{BQP_PPAD}
\bibliographystyle{plain}

\newpage
\thispagestyle{empty}

\appendix

\section{Proof of Theorem \ref{thm: QE to CE}}
Recall that we are given that $\mu _i(\rho) \ge \mu _i(\Phi_i(\rho) )$ for all players $i$ and all admissible super-operators $\Phi_i$ on $H_i$, and we want to prove that for all players $i$ and all strategies $s_i, s'_i\in S_i$,
\begin{equation}\label{eq: CNE}
	\sum_{s_{-i}} p(s_i,s_{-i}) u_i(s_i,s_{-i}) \geq  \sum_{s_{-i}} p(s_i,s_{-i}) u_i(s_i',s_{-i})
\end{equation}
for $p(s) = \rho_{ss}$.

Fix $i$ and $s_i, s'_i$. Consider the admissible super-operator $\Phi_i$ defined by
\begin{equation}
	\Phi_i = \sum_{t_i \neq s_i} P_{t_i} \rho P_{t_i} + (s_i \leftrightarrow s_i') P_{s_i} \rho P_{s_i} (s_i \leftrightarrow s_i')
\end{equation}
where $P_{t_i}$ is the projection onto the subspace $span(t_i)\otimes H_{-i}$, and $(s_i \leftrightarrow s_i')$ is the operator swapping $s_i$ and $s_i'$. It is not hard to verify that $\Phi_i$ is an admissible super-operator. Next we will show that the difference of $\mu_i(\rho)$ and $\mu_i(\Phi_i(\rho))$ is the same as that of the two sides of Eq. \eqref{eq: CNE}.

\begin{align}
\nonumber \mu_i(\rho) & = \av[u_i(s(\rho))]\\
\nonumber & = \sum_{\bar s\in S} \bra{\bar s}\rho \ket{\bar s} u_i(\bar s) = \sum_{\bar s\in S} p(\bar s) u_i(\bar s) \\
 & = \sum_{\bar s_i \neq s_i} \sum_{\bar s_{-i}} p(\bar s)u_i(\bar s) + \sum_{\bar s_{-i}} p(s_i\bar s_{-i})u_i(s_i\bar s_{-i})
\end{align}

\begin{align}
\nonumber\mu_i(\Phi_i(\rho)) & = \sum_{\bar s\in S} \bra{\bar s} \Phi_i(\rho) \ket{\bar s} u_i(\bar s) \\
\nonumber& = \sum_{\bar s\in S} \bra{\bar s} \sum_{t_i \neq s_i} P_{t_i} \rho P_{t_i} + (s_i \leftrightarrow s_i') P_{s_i} \rho P_{s_i} (s_i \leftrightarrow s_i') \ket{\bar s} u_i(\bar s) \\
\nonumber& = \sum_{\bar s\in S} \bra{\bar s} \sum_{t_i \neq s_i} P_{t_i} \rho P_{t_i} \ket{\bar s} u_i(\bar s) + \sum_{\bar s\in S} \bra{\bar s} (s_i \leftrightarrow s_i') P_{s_i} \rho P_{s_i} (s_i \leftrightarrow s_i') \ket{\bar s} u_i(\bar s) \\
& = \sum_{t_i \neq s_i}\sum_{\bar s_{-i}} p(t_i \bar s_{-i}) u_i(t_i \bar s_{-i}) + \sum_{\bar s_{-i}} p(s_i \bar s_{-i}) u_i(s'_i\bar s_{-i}) \label{eqn:1}
\end{align}

Since $\rho$ is a quantum correlated equilibrium, we have $\mu_i(\rho) \geq \mu_i(\Phi_i(\rho))$. Comparing the above two expressions for $\mu_i(\rho)$ and $\mu_i(\Phi_i(\rho))$ gives Eq. \eqref{eq: CNE} as desired.

\section{Proof of Theorem \ref{thm: CE to QE}}
\thispagestyle{empty}

Let $ \alpha(s_i) = \sum_{s_{-i}} p(s_i,s_{-i}) u_i(s_i,s_{-i})$ and $\beta(s_i, s_i') =  \sum_{s_{-i}} p(s_i,s_{-i}) u_i(s_i',s_{-i})$. Now for any $i$, we have
\begin{align}\label{eq: original payoff}
\mu_i(\rho) =  \sum_s \bra{s}\rho \ket{s} u_i(s) =  \sum_s p(s) u_i(s) = \sum_{s_i} \sum_{s_{-i}} p(s_i s_{-i}) u_i(s_i  s_{-i}) =  \sum_{s_i} \alpha(s_i)
\end{align}
where the first two steps are by the definition of $\mu_i$ and $p$. Now for an arbitrary TPCP super-operator $\Phi_i$, we use its Kraus representation to obtain
\begin{align}
	\Phi_i(\rho) = \sum_{j=1}^k (A_{ij} \otimes I_{-i})\rho(A_{ij}^* \otimes I_{-i})
\end{align}
with constraint $\sum_{j=1}^k A_{ij}^* A_{ij} = I_i$, where $I_i$ is the identity super-operator from $L(H_i)$ to $L(H_i)$. Now we have
\begin{align*}
\mu_i(\Phi_i(\rho)) &= \sum_{s'} \bra{s'} \Phi _i(\rho) \ket{s'} u_i(s') \qquad // \text{ by the def of } \mu_i \\
&=  \sum_{s'} \bra{s'}\sum_{j=1}^k (A_{ij} \otimes I_{-i})\rho(A_{ij}^* \otimes I_{-i}) \ket{s'} u_i(s') \\
&=  \sum_{s'} \sum_{j=1}^k  \bra{s'}(A_{ij} \otimes I_{-i})(\sum_s p(s) \ket{s} \bra{s})(A_{ij}^* \otimes I_{-i}) \ket{s'} u_i(s') \qquad // \text{ by the def of } \rho \\
&=  \sum_{s'} \sum_s \sum_{j=1}^k \bra{s'} A_{ij} \otimes I_{-i}\ket{s} \bra{s}A_{ij}^* \otimes I_{-i} \ket{s'} p(s) u_i(s') \\
&=  \sum_{s'_i} \sum_{s} \sum_{j=1}^k \bra{s'_i} A_{ij}\ket{s_i} \bra{s_i}A_{ij}^*\ket{s_i'}  p(s_is_{-i}) u_i(s'_is_{-i}) \\
&=  \sum_{s'_i} \sum_{s_i} \sum_{j=1}^k \bra{s'_i} A_{ij}\ket{s_i} \bra{s_i}A_{ij}^*\ket{s_i'}  \beta(s_i, s_i') \qquad // \text{ by the def of } \beta(s_i, s_i') \\
\end{align*}
Note that $\bra{s'_i} A_{ij}\ket{s_i} \bra{s_i}A_{ij}^*\ket{s_i'} = \|\bra{s'_i} A_{ij}\ket{s_i} \|^2 \geq 0$, thus by the assumption that $\beta(s_i, s_i') \leq \alpha(s_i')$ (i.e. $p$ is a classical correlated equilibrium), we have
\begin{align*}
\mu_i(\Phi_i(\rho)) &\leq \sum_{s_i} \sum_{j=1}^k  \sum_{s'_i} \bra{s'_i} A_{ij}\ket{s_i} \bra{s_i}A_{ij}^*\ket{s_i'}  \alpha(s_i) \\
&=  \sum_{s_i} \sum_{j=1}^k \bra{s_i} A_{ij}^*  (\sum_{s'_i}\ket{s'_i} \bra{s'_i}) A_{ij}\ket{s_i}  \alpha(s_i) \\
&=  \sum_{s_i} \bra{s_i} \sum_{j=1}^k A_{ij}^*  A_{ij} \ket{s_i}  \alpha(s_i') \\
&=  \sum_{s_i} \qip{s_i}{s_i}  \alpha(s_i) \\
&=  \mu_i(\rho)
\end{align*}
where the last equality is by Eq. \eqref{eq: original payoff}. This completes the proof of Theorem \ref{thm: CE to QE}.

\subsection{Proof of Corollary \ref{thm: QNE-exist}}
\thispagestyle{empty}

We will reduce the existence of a quantum Nash equilibrium to the existence of a Nash equilibrium.

For a given game $G$ with finite players and finite strategies, there always exists a Nash equilibrium, say $p$. We transform $p$ into a quantum state $\rho$ using the the mapping $\rho = \sum_s p(s)\ket{s}\bra{s}$. By Theorem \ref{thm: CE to QE}, $\rho$ is guaranteed to be quantum Nash equilibrium of $G$.

Thus, quantum Nash equilibria always exist.


\section{Proof of Theorem \ref{thm: purestate}}

\subsection{Examples of the First Statement}
Define utility functions of Player 1 and 2 to be:

\[
	A = \begin{bmatrix} 270 & 126 \\ 0 & 270 \end{bmatrix}.
\]

Suppose the initial state is

\begin{align*}
\ket{\psi} & = \sqrt{1/3}\ket{00}+\sqrt{1/6}\ket{01}+\sqrt{1/6}\ket{10}+\sqrt{1/3}\ket{11},
\end{align*}

whose corresponding density matrix is

\begin{align*}
	\rho & =
	\begin{bmatrix}
		1/3& \sqrt{1/18}& \sqrt{1/18}& 1/3 \\
        \sqrt{1/18}& 1/6&1/6 & \sqrt{1/18} \\
        \sqrt{1/18}& 1/6& 1/6& \sqrt{1/18} \\
        1/3& \sqrt{1/18}& \sqrt{1/18}&1/3
	\end{bmatrix},
\end{align*}

and whose corresponding classical correlated distribution is
\begin{align*}
 p = \begin{bmatrix} 1/3 & 1/6\\1/6  & 1/3	\end{bmatrix},
\end{align*} \\
which is easily verified to be a classical correlated equilibrium.

However, $\rho$ is not a quantum Nash equilibrium. Define a unitary matrix
\[
	G = \begin{bmatrix} \sqrt{2/3} & \sqrt{1/3} \\ \sqrt{1/3} & -\sqrt{2/3} \end{bmatrix}.
\]
Consider
\[
	\rho' = (G\otimes I)\rho (G\otimes I) =
	\begin{bmatrix}
	1/2 & \sqrt{2}/3 & 0 & -1/6\\
	\sqrt{2}/3 & 4/9 & 0 & -\sqrt{2}/9   \\
	0& 0& 0&0    \\
	-1/6& -\sqrt{2}/9   & 0& 1/18
	\end{bmatrix}.
\]
It is easily seen that $\rho'$ has higher expected utility value for player 1, actually
\[
	\mu_1(\rho') = 206, \qquad \mu_1(\rho) = 201.
\]

\subsection{Proof of the Second Statement}
\thispagestyle{empty}

Let $\rho = \ket{\psi}\bra{\psi} = \sum_{a,b} \sqrt{p(a)p(b)}\ket{a}\bra{b}$. Then
\begin{align*}\label{eq: eg1-rho}
\mu_i(\rho)
& =  \sum_s \bra{s}\rho \ket{s} u_i(s) \\
& =  \sum_s \bra{s}\sum_{a,b} \sqrt{p(a)p(b)}\ket{a}\bra{b} \ket{s} u_i(s)\\
& =  \sum_{s} p(s) u_i(s)\\
& =  \sum_{s_i} p_i(s_i)\sum_{s_{-i}}p_{-i}(s_{-i}) u_i(s_i, s_{-i})\\
& =  \sum_{s_i: p_i(s_i)>0} p_i(s_i)\sum_{s_{-i}}p_{-i}(s_{-i}) u_i(s_i, s_{-i})
\end{align*}
Now assume that Player $i$ applies an admissible super-operator $\Phi_i$ on $\rho$:
\begin{align*}
	\Phi_i(\rho) =  \sum_{j=1}^k (A_{ij} \otimes I_{-i})\rho(A_{ij}^* \otimes I_{-i})
\end{align*}
where $\sum_{j=1}^k A_{ij}^* A_{ij} = I_i$. \\
Let $\bar s_i$ be a strategy \st $p_i(\bar s_i) > 0$. Then by the definition of Nash equilibrium, we have \begin{equation}
	\sum_{s_i} p_{-i}(s_{-i})u_i(s_i s_{-i}) \leq \sum_{s_i} p_{-i}(s_{-i})u_i(\bar s_i s_{-i}),
\end{equation}
for any $s_i$.
\begin{align*}
\mu_i(\Phi_i(\rho))
& = \sum_{s} \bra{s} \Phi _i(\rho) \ket{s} u_i(s) \\
& = \sum_{s} \bra{s}  \sum_{j=1}^k (A_{ij} \otimes I_{-i})\rho(A_{ij}^* \otimes I_{-i}) \ket{s} u_i(s) \\
& = \sum_{s} \bra{s}  \sum_{j=1}^k (A_{ij} \otimes I_{-i})\sum_{a,b} \sqrt{p(a)p(b)}\ket{a}\bra{b}(A_{ij}^* \otimes I_{-i}) \ket{s} u_i(s) \\
& = \sum_{s,a,b,j} \sqrt{p(a)p(b)}  \bra{s} (A_{ij} \otimes I_{-i}) \ket{a}\bra{b}(A_{ij}^* \otimes I_{-i}) \ket{s} u_i(s) \\
& = \sum_{s,a,b,j} \sqrt{p_i(a_i)p_i(b_i)} \sqrt{p_{-i}(a_{-i})p_{-i}(b_{-i})} \bra{s_i}A_{ij}\ket{a_i} \qip{s_{-i}}{a_{-i}}\bra{b_i}A_{ij}^* \ket{s_i}\qip{b_{-i}}{s_{-i}} u_i(s) \\
& = \sum_{s_i,s_{-i},a_i,b_i,j} \sqrt{p_i(a_i)p_i(b_i)}\bra{s_i}A_{ij}\ket{a_i}\bra{b_i}A_{ij}^* \ket{s_i}p_{-i}(s_{-i})u_i(s_i,s_{-i})\\
& = \sum_{s_i,s_{-i},a_i,b_i,j:p_i(a_i)>0,p_i(b_i)>0} \sqrt{p_i(a_i)p_i(b_i)}\bra{s_i}A_{ij}\ket{a_i}\bra{b_i}A_{ij}^* \ket{s_i}p_{-i}(s_{-i})u_i(s_i,s_{-i})\\
& = \sum_{s_i,a_i,b_i,j:p_i(a_i)>0,p_i(b_i)>0} \sqrt{p_i(a_i)p_i(b_i)}\bra{s_i} A_{ij}\ket{a_i}\bra{b_i}A_{ij}^*\ket{s_i}\sum_{s_{-i}}p_{-i}(s_{-i})u_i(s_i,s_{-i})\\
& \le  \sum_{s_i,a_i,b_i,j:p_i(a_i)>0,p_i(b_i)>0} \sqrt{p_i(a_i)p_i(b_i)}\bra{s_i} A_{ij}\ket{a_i}\bra{b_i}A_{ij}^*\ket{s_i}\sum_{s_{-i}}p_{-i}(s_{-i})u_i(a_i,s_{-i})\\
\end{align*}
\begin{align*}
& =  \sum_{s_i,a_i,b_i,j:p_i(a_i)>0,p_i(b_i)>0} \sqrt{p_i(a_i)p_i(b_i)}\bra{b_i}A_{ij}^*\ket{s_i}\bra{s_i} A_{ij}\ket{a_i}\sum_{s_{-i}}p_{-i}(s_{-i})u_i(a_i,s_{-i})\\
& =  \sum_{a_i,b_i,j:p_i(a_i)>0,p_i(b_i)>0} \sqrt{p_i(a_i)p_i(b_i)}\bra{b_i}A_{ij}^* A_{ij}\ket{a_i}\sum_{s_{-i}}p_{-i}(s_{-i})u_i(a_i,s_{-i})\\
& =  \sum_{a_i,b_i: p_i(a_i)>0,p_i(b_i)>0} \sqrt{p_i(a_i)p_i(b_i)}\qip{b_i}{a_i}\sum_{s_{-i}}p_{-i}(s_{-i})u_i(a_i,s_{-i})\\
& =  \sum_{\bar s_i: p_i(\bar s_i)>0} p_i(\bar s_i)\sum_{s_{-i}}p_{-i}(s_{-i}) u_i(\bar s_i, s_{-i})\\
& = \mu_i(\rho)
\end{align*}
This completes the proof of Theorem \ref{thm: purestate}.

From the above proof, one can see that if $supp(p_i) = S_i$, then the only inequality becomes the equality. We thus obtain the following fact.
\thispagestyle{empty}

\begin{Cor}\label{thm: specialpure}
If $p$ is a classical Nash equilibrium and $supp(p_i) = S_i$, then $\ket{\psi} = \sum_{s} \sqrt{p(s)}\ket{s}$ is a quantum Nash equilibrium, and any quantum operation by Player $i$ does not change his/her utility value.
\end{Cor}

\section{Examples in Theorem \ref{thm: examples}}
\thispagestyle{empty}

Define the utility matrices of both Player 1 and Player 2 to be:
\[
	u_1 = u_2 = u = \begin{bmatrix} 2 & 1 \\ 1 & 2 \end{bmatrix}
\]
Note that since $u$  is symmetric, so is the game. Below we will show a couple of examples where $p_\rho$ is a classical (sometimes correlated) Nash equilibrium but $\rho$ itself is not a quantum (correlated) Nash equilibrium.

\subsection*{Example 1: a mixed product state}
\thispagestyle{empty}

Suppose the initial state is
\begin{align}\label{eq: eg1-rho}
	\rho & = \frac{1}{2} \begin{bmatrix} \cos^2(\theta) & \cos(\theta)\sin(\theta) \\ \cos(\theta)\sin(\theta)  & \sin^2(\theta)	 \end{bmatrix} \otimes \ket{0}\bra{0}
	+ \frac{1}{2}\begin{bmatrix} \sin^2(\theta) & -\cos(\theta)\sin(\theta) \\ -\cos(\theta)\sin(\theta)  & \cos^2(\theta)	 \end{bmatrix}\otimes \ket{1}\bra{1}
	\\
	\nonumber \\
	& =
	\begin{bmatrix}
		\cos^2(\theta)/2 &      & \cos(\theta)\sin(\theta)/2&     \\
  		    & \sin^2(\theta)/2 &      & -\cos(\theta)\sin(\theta)/2    \\
		\cos(\theta)\sin(\theta)/2 &      & \sin^2(\theta)/2 &     \\
  		    & -\cos(\theta)\sin(\theta)/2     &      & \cos^2(\theta)/2
	\end{bmatrix}
\end{align}
Take the diagonal elements to form a classical correlated distribution
\[
	p = \begin{bmatrix} \cos^2(\theta)/2 & \sin^2(\theta)/2 \\ \sin^2(\theta)/2  & \cos^2(\theta)/2	 \end{bmatrix},
\]
which is easily verified to be a classical correlated equilibrium if $\cos^2(\theta) \geq 1/2$. \\

However, $\rho$ is not a quantum Nash equilibrium. Define a unitary matrix
\[
	G = \begin{bmatrix} \cos(\theta) & \sin(\theta) \\ \sin(\theta) & -\cos(\theta) \end{bmatrix}.
\]
Consider
\[
	\rho' = (G\otimes I)\rho (G\otimes I) =
	\begin{bmatrix}
	1/2 & & & \\
	& 0& &    \\
	& & 0&    \\
	& & & 1/2
	\end{bmatrix}
\]
It is easily seen that $\rho'$ has higher expected utility value for player 1, actually
\[
	\mu_1(\rho') = 2, \qquad \mu_1(\rho) = 1+\cos^2(\theta).
\]

\subsection*{Example 2: an entangled pure state}
Consider
\[
	\rho = \frac{1}{2}
	\begin{bmatrix}
		\cos^2(\theta) & \cos(\theta)\sin(\theta)& \cos(\theta)\sin(\theta)& -\cos^2(\theta)    \\
		\cos(\theta)\sin(\theta) & \sin^2(\theta) & \sin^2(\theta) & -\cos(\theta)\sin(\theta)    \\
		\cos(\theta)\sin(\theta) & \sin^2(\theta) & \sin^2(\theta) & -\cos(\theta)\sin(\theta)    \\
		-\cos^2(\theta)& -\cos(\theta)\sin(\theta)&-\cos(\theta)\sin(\theta)&\cos^2(\theta)
	\end{bmatrix}
\]
Since the diagonal entries are the same as those in Eq. \eqref{eq: eg1-rho}, the induced classical distribution is also the same as before, which is a classical correlated equilibrium. Again, $\rho$ is not a quantum Nash equilibrium since
\[
	\rho' = (G\otimes I)\rho (G\otimes I) =
	\begin{bmatrix}
		1/2 & & &1/2 \\
		& 0& &    \\
		& & 0&    \\
		1/2& & & 1/2
	\end{bmatrix}
\]
and it is easy to see that $\mu_1(\rho') = 2$.

\subsection*{Example 3: (uncorrelated) Nash equilibrium}
Suppose
\[
	\rho =
	\begin{bmatrix}
		1/4 & 1/4 & 1/4 & -1/4    \\
		1/4 & 1/4 & 1/4 & -1/4    \\
		1/4 & 1/4 & 1/4 & -1/4    \\
		-1/4& -1/4&-1/4 & 1/4
	\end{bmatrix}
\]
The induced classical distribution is now
\[
	p = \begin{bmatrix} 1/4 & 1/4 \\ 1/4 & 1/4 \end{bmatrix}.
\]
It is easy to check that this is a classical correlated equilibrium. Consider
\[
	\rho' =(H\otimes I)\rho (H\otimes I) =
	\begin{bmatrix}
		1/2 & & &1/2 \\
		& 0& &    \\
		& & 0&    \\
		1/2& & & 1/2
	\end{bmatrix}
\]
where H is the Hadamard matrix. Here $\mu_1(\rho') = 2 > \mu_1(\rho)$. Therefore $\rho$ is not a quantum Nash equilibrium.

\section{Examples for Theorem \ref{thm: examples2}}
\thispagestyle{empty}

Define the utility matrices of both Player 1 and Player 2 to be:
\begin{equation}
	u_1 = u_2 = u = \begin{bmatrix} 2 & 1 \\ 1 & 2 \end{bmatrix}
\end{equation}
Consider the classical correlated distribution
\begin{equation}
	p = \begin{bmatrix} 1/2 & 0 \\ 0  & 1/2	\end{bmatrix},
\end{equation}
It is easy to check that this is a classical correlated equilibrium.

For any quantum state $\rho$ with $\rho_{ss} = p(s)$, the expected utility value for player 1 is given by $\mu_1(\rho') = 2$.  It is impossible to have any density operator $\rho'$ with $\mu_1(\rho') > 2$. It is easy to see that the expected utility value is maximized so $\rho$ is a quantum Nash equilibrium.

\end{document}